\documentclass[11pt]{article}

\usepackage[english]{babel}

\usepackage[margin=1in]{geometry}

\usepackage{graphicx}
\usepackage{amsfonts,amsmath,amssymb,amsthm} 
\usepackage{comment}
\usepackage{xcolor}
\definecolor{ForestGreen}{rgb}{0.1333,0.5451,0.1333}
\definecolor{DarkRed}{rgb}{0.80,0,0}
\definecolor{Red}{rgb}{1,0,0}
\usepackage[linktocpage=true,
pagebackref=true,colorlinks,
linkcolor=DarkRed,citecolor=ForestGreen,
bookmarks,bookmarksopen,bookmarksnumbered]
{hyperref}
\usepackage{cleveref}
\usepackage{algorithm}
\usepackage[noend]{algpseudocode}
\usepackage{booktabs}
\usepackage{multirow}

\theoremstyle{plain}
\newtheorem{theorem}{Theorem}[section]
\newtheorem{lemma}[theorem]{Lemma}

\newtheorem{corollary}[theorem]{Corollary}

\theoremstyle{definition}
\newtheorem{definition}[theorem]{Definition}

\newcommand{\R}{\mathbb{R}}

\newcommand{\calM}{\mathcal{M}}

\newcommand{\bfx}{\mathbf{x}}
\newcommand{\bfv}{\mathbf{v}}

\newcommand{\bfp}{\mathbf{p}}

\newcommand{\gic}{\mathrm{GIR}}
\newcommand{\gics}{\mathrm{SGIR}}
\newcommand{\ic}{\mathrm{IR}}
\newcommand{\rr}{\mathrm{RR}}
\newcommand{\ps}{\mathrm{PS}}
\newcommand{\mnw}{\mathrm{MNW}}
\DeclareMathOperator*{\argmax}{argmax}

\floatname{algorithm}{Mechanism}

\begin{document}

\title{Incentive Analysis of Collusion in Fair Division}

\author{
    Haoqiang Huang\thanks{HKUST. Email: \texttt{haoqiang.huang@connect.ust.hk}}
    \and
    Biaoshuai Tao\thanks{Shanghai Jiao Tong University. Email: \texttt{bstao@sjtu.edu.cn}}
    \and
    Mingwei Yang\thanks{Stanford University. Email: \texttt{mwyang@stanford.edu}}
    \and
    Shengwei Zhou\thanks{University of Macau. Email: \texttt{yc17423@um.edu.mo}}
}

\date{}

\maketitle

\begin{abstract}
    We study fair division problems with strategic agents capable of gaining advantages by manipulating their reported preferences.
    Although several impossibility results have revealed the incompatibility of truthfulness with standard fairness criteria, subsequent works have circumvented this limitation through the \emph{incentive ratio} framework.
    Previous studies demonstrate that fundamental mechanisms like \emph{Maximum Nash Welfare (MNW)} and \emph{Probabilistic Serial (PS)} for divisible goods, and \emph{Round-Robin (RR)} for indivisible goods achieve an incentive ratio of $2$, implying that no individual agent can gain more than double his truthful utility through manipulation. 
    However, collusive manipulation by agent groups remains unexplored.
    
    In this work, we define \emph{strong group incentive ratio (SGIR)} and \emph{group incentive ratio (GIR)} to measure the gain of collusive manipulation, where SGIR and GIR are respectively the maximum and minimum of the incentive ratios of corrupted agents.
    Then, we tightly characterize the SGIRs and GIRs of MNW, PS, and RR.
    In particular, the GIR of MNW is $2$ regardless of the coalition size.
    Moreover, for coalition size $c \geq 1$, the SGIRs of MNW and PS, and the GIRs of PS and RR are $c + 1$.
    Finally, the SGIR of RR is unbounded for coalition size $c \geq 2$.
    Our results reveal fundamental differences of these three mechanisms in their vulnerability to collusion.
\end{abstract}

\section{Introduction}

Fair division is a fundamental and highly applicable field of study within economics and computer science.
The core challenge is to assign a set of resources among interested agents in a way that is deemed fair.
A typical model includes $m$ goods and $n$ competing agents, and asks for an allocation of goods among agents.
The fairness of an allocation is evaluated via the preferences of agents over the goods, which are often private information.
Depending on whether the goods are positively valued and whether a good can be partially assigned (i.e., divisible) or not (i.e., indivisible), this simple model conceives a variety of interesting settings that capture the scenarios of property division, assigning limited medical resources during a public health crisis, allocating chores to cleaning staff, and so on.

Since the unfair division of resources can lead to a series of negative consequences like conflicts and inefficiency, fair division has attracted extensive research interest for decades. The most prevalent criterion of fairness is \emph{envy-freeness (EF)}. An allocation is EF if any agent weakly prefers his allocated set of resources to that of any other. However, an EF allocation may not exist for indivisible goods due to the simple example where two agents compete for one good. Hence, several relaxations of EF such as \emph{envy-freeness up to one good (EF1)}~\cite{DBLP:conf/sigecom/LiptonMMS04,DBLP:conf/bqgt/Budish10} and \emph{envy-freeness up to any good (EFX)}~\cite{DBLP:conf/ecai/GourvesMT14,DBLP:journals/teco/CaragiannisKMPS19} have been proposed. We refer to the recent excellent surveys~\cite{DBLP:conf/ijcai/AmanatidisBFV22,DBLP:journals/sigecom/AzizLMW22} for more details.

Over the years, several mechanisms for fair division have been studied, including the prominent \emph{Round-Robin (RR)}~\cite{tao2024fairtruthfulmechanismsadditive}, \emph{Probabilistic Serial (PS)}~\cite{DBLP:journals/jcss/HuangWWZ24,DBLP:conf/wine/JiangNVW23}, and \emph{Maximum Nash Welfare (MNW)} mechanisms~\cite{DBLP:journals/teco/CaragiannisKMPS19,DBLP:journals/ai/BeiTWY25}. 
It is widely known that PS and MNW satisfy EF for divisible goods, while RR and MNW satisfy EF1 for indivisible goods.

However, these fairness guarantees of classic mechanisms are built on the assumption that all participated agents will report their private preferences over the goods truthfully.
It is natural that agents may not simply trust the mechanisms and report their preferences honestly.
Moreover, it has been shown that if the agents are not truthful, the resulting allocation will change drastically and no longer be fair~\cite{DBLP:journals/ijgt/EkiciK16}.
These consequences give rise to the attention on the \emph{truthful} aspect of fair division mechanisms and motivate the game-theoretic perspective of the fair division problem, under which the pursuit of fairness becomes considerably less tractable.
A line of literature has focused on designing mechanisms that are both truthful and (approximately) fair~\cite{DBLP:conf/sigecom/LiptonMMS04,DBLP:conf/aldt/CaragiannisKKK09,DBLP:journals/mp/AzizLW24}.
Unfortunately, it is shown that truthfulness is in general incompatible with natural fairness notions, except when monetary transfer is allowed~\cite{DBLP:conf/sigecom/AmanatidisBCM17,DBLP:journals/ai/BuST23}.

Given the impossibility of satisfying fairness and truthfulness simultaneously, a compromise is to preserve fairness guarantees while relaxing strict truthfulness requirements.
The notion of \emph{incentive ratio}, as a natural relaxation of truthfulness, has been widely studied in the context of Fisher markets~\cite{DBLP:journals/iandc/ChenDTZZ22} and resource sharing~\cite{DBLP:conf/ipps/ChengDL20,DBLP:conf/sigecom/ChengDLY22,DBLP:journals/dam/ChenCDQY19}, and is recently introduced to the fair division problem~\cite{DBLP:journals/jcss/HuangWWZ24,DBLP:journals/ai/BeiTWY25,DBLP:conf/aaai/XiaoL20,tao2024fairtruthfulmechanismsadditive}. Moreover, the definition of incentive ratio is also closely related to the concept of \emph{approximate Nash equilibrium}, an extensively studied concept in game theory.

Generally speaking, the incentive ratio of a mechanism is defined as the worst-case ratio between the utilities achieved by an agent under manipulated preference and his truthful preference.
This metric evaluates the  incentive to deviate through a worst-case analysis, reflecting the most extreme scenarios where agents could benefit from dishonesty. Previous studies show that RR admits an incentive ratio of $2$ for indivisible goods~\cite{DBLP:conf/aaai/XiaoL20,tao2024fairtruthfulmechanismsadditive}, and PS and MNW both achieve an incentive ratio of 2 for divisible goods~\cite{DBLP:journals/iandc/ChenDTZZ22,DBLP:journals/jcss/HuangWWZ24,DBLP:conf/atal/0037SX24,DBLP:journals/ai/BeiTWY25}.
Moreover, for indivisible goods, incentive ratio lower bounds that lie between $(1, 2)$ are established under various valuation function classes and fairness criteria~\cite{tao2024fairtruthfulmechanismsadditive}.

The incentive ratio framework is useful in the sense that the mechanism with a lower incentive ratio is more likely to achieve truthfulness when agents are sufficiently rational.
There are two reasons.
First, the analysis of the incentive ratio follows the worst-case style: it assumes that each agent has both the comprehensive knowledge of other agents’ preferences—often unavailable in practice—and the capability of getting solutions to computationally prohibitive optimization problems for computing the optimal manipulated preference; yet, in realistic scenarios, performing optimal manipulation is costly.
Second, the ``perfect rationality'' assumption that agents are consistently rational and always act in a way that maximizes the utility is also insufficient to capture human decision-making's complex ethical and behavioral dimensions. The incentive ratio framework, playing a role in the middle of both sides, captures the individual behaviors under bounded rationality. Hence, some constant incentive ratio, even like 2, implies that the agents' incentive to manipulate the mechanism is reasonably bounded~\cite{DBLP:journals/jcss/HuangWWZ24}.

However, all the previous studies on incentive ratio neglect the case where a group of agents can jointly misreport their preferences to gain more manipulation capability against the fair division mechanisms, which is not uncommon.
To characterize agents' incentive to form a coalition to conduct strategic behaviors, we consider natural generalizations of incentive ratio and analyze the robustness of classic mechanisms against collusive manipulation.

\subsection{Our Contributions}

When agents can collectively manipulate the mechanisms, we introduce natural generalizations of incentive ratio that we call \emph{strong group incentive ratio (SGIR)} and \emph{group incentive ratio (GIR)} (Definitions~\ref{def:strong-gir} and~\ref{def:gir}), which closely resembles the well-studied notions of \emph{(weak) group strategyproofness} \cite{DBLP:conf/wine/AzizY14} and \emph{(super) strong Nash equilibrium}~\cite{DBLP:journals/ijgt/BraggionGLSS20,DBLP:conf/sigecom/HanSTX23} (see the discussion below \Cref{cor:compare-defs-gic} for more detailed comparisons).
For a fixed coalition size $c \geq 1$, we allow at most $c$ agents to be corrupted and form a coalition, and they can decide collectively how to misreport their valuation functions.
Informally, the SGIR of a mechanism is the smallest $R \geq 1$, which may depend on $c$, such that misreporting cannot increase the utility of at least one corrupted agent by a factor strictly larger than $R$ while guaranteeing that every corrupted agent is weakly better off.
Moreover, the GIR of a mechanism is the smallest $R \geq 1$ such that misreporting cannot increase the utility of every corrupted agent by a factor strictly larger than $R$.
Note that the SGIR of a mechanism always upper bounds its GIR, and both SGIR and GIR are equivalent to incentive ratio when $c = 1$.

\begin{table}[t]
\centering
\begin{tabular}{@{}llll@{}}
\toprule
                             & Maximum Nash Welfare & Probabilistic Serial & Round-Robin \\ \midrule
SGIR & $c + 1$              & $c + 1$              & $+\infty$  \\
GIR        & $2$                   & $c + 1$              & $c + 1$     \\ \bottomrule
\end{tabular}
\caption{Our results for the SGIRs and GIRs of MNW, PS, and RR with coalition size $c \geq 1$. In particular, the unbounded SGIR of RR holds only for $c \geq 2$. Both SGIR and GIR are equivalent to incentive ratio when $c = 1$, and it is previously known that the incentive ratios of all three mechanisms are $2$~\cite{DBLP:conf/aaai/XiaoL20,DBLP:journals/jcss/HuangWWZ24,DBLP:journals/ai/BeiTWY25}.}
\label{fig:results}
\end{table}

As our main contributions, which are summarized in \Cref{fig:results}, we tightly characterize the SGIRs and GIRs of the MNW, PS, and RR mechanisms, where we assume goods are divisible for MNW and PS, and are indivisible for RR.
Notably, our results in the special case of $c = 1$ generalize the incentive ratio $2$ of all three mechanisms.

We first show that the SGIR and GIR of MNW for coalition size $c \geq 1$ are $c + 1$ and $2$, respectively (\Cref{thm:gic-mnw}).
We achieve this by leveraging the connection between MNW and \emph{Fisher markets}: the set of allocations maximizing Nash welfare coincides with the set of allocations corresponding to \emph{market equilibria} when all agents have identical budgets.
Then, we generalize the proof of \cite[Theorem 3.1]{DBLP:journals/iandc/ChenDTZZ22}, which shows that the incentive ratio for market equilibrium is $2$ when the utility functions satisfy the \emph{weak gross substitute} property, to establish the SGIR and GIR of MNW.

To establish our results for PS and RR, we first show that the SGIR and GIR of PS are upper bounded respectively by the SGIR and GIR of RR (\Cref{thm:gics-ps-grr}).\footnote{This also implies that the incentive ratio of PS is upper bounded by that of RR, which provides an alternative proof of the incentive ratio upper bound of $2$ for PS.}
This is achieved by reducing PS to RR, in the sense that the allocation returned by PS on an instance with divisible goods can be computed by RR on a transformed instance with indivisible goods (\Cref{thm:equi-ps-rr}).
We note that the reduction from PS to RR might seem unsurprising in hindsight, given that PS can be viewed as a continuous version of RR, and the reduction works simply by creating sufficiently many indivisible copies of each divisible good.
Nevertheless, to the best of our knowledge, this connection has not been formalized in the past literature.
In fact, the possibility of transforming a PS problem into an RR problem is recently raised as an open problem by~\cite{DBLP:conf/atal/0037SX24}.

To further illustrate the non-trivial nature of the reduction from PS to RR, first observe that the number of indivisible copies for each divisible good, denoted as $T$, has to be a multiple of $n$, as witnessed by the instance where all agents have identical ordinal preferences.
As agents' preferences diverge, the execution of PS becomes much more involved, which introduces additional constraints on the divisibility of $T$.
By carefully coupling the executions of RR and a combinatorial implementation of PS, we show that $T = (n!)^m$ meets all the divisibility constraints.

Then, we show that for coalition size $c \geq 1$, the GIR of RR is upper bounded by $c + 1$ (\Cref{thm:gic-rr-ub}), and the GIR of PS is lower bounded by $c + 1$ (\Cref{thm:gic-ps-lb}).
It follows by the connection between PS and RR that both bounds are tight (\Cref{cor:lb-gic-grr-eve}).
For the SGIR of PS, although it is not directly implied by the GIR of RR, we manage to show that the SGIR of PS is upper bounded by $c + 1$ for coalition size $c \geq 1$ by leveraging the ingredients in the proofs of the above connection and RR's GIR upper bound (\Cref{thm:gics-ps-ub}).
Finally, we show that the SGIR of RR is unbounded for coalition size $c \geq 2$ (\Cref{thm:gic-rr}).

Although MNW, PS, and RR all admit an incentive ratio of $2$, our results reveal that they exhibit varied vulnerability when considering collective manipulation.
In particular, these results strictly separate MNW, PS, and RR in terms of their resistance against strategic behaviors with MNW greatly outperforming the other two, and hence further substantiate the prominence of MNW as evidenced by prior work~\cite{DBLP:conf/wine/0002PP020,DBLP:journals/ai/BeiTWY25}.

\subsection{Related Work}

The existence of fair and truthful mechanisms has been extensively studied in various settings.
When goods are indivisible, \cite{DBLP:conf/sigecom/AmanatidisBCM17} characterize all truthful mechanisms for two agents and show that truthfulness is incompatible with any meaningful fairness notion.
For divisible goods, fairness and truthfulness can be easily achieved by simply allocating each good evenly to all agents.
Moreover, \cite{DBLP:conf/sigecom/ColeGG13} present a truthful mechanism that achieves a $1/e$-approximation guarantee to the Nash welfare maximizing solution.
When each good can be either divisible or indivisible, \cite{Li2023Truthful} show that truthfulness and fairness are incompatible even for instances with one indivisible good and one divisible good, and the compatibility only holds for very restrictive special cases of valuation functions.
\cite{DBLP:conf/sagt/LiTWWYZ25} extend a number of results for truthfully allocating goods to the setting where items are negatively valued, i.e., chores.

Another widely considered setting with divisible resources is referred to as \emph{cake-cutting}, where the cake is modeled by an interval $[0,1]$, and each agent's preference is described by a \emph{value density function}.
In the \emph{Robbertson-Webb} model, where mechanisms communicate with agents via queries to gradually learn their preferences, there exist strong impossibility results regarding the existence of fair and truthful mechanisms~\cite{kurokawa2013cut,branzei2015dictatorship}.
In a more tractable model, all agents report their value density functions to mechanisms simultaneously.
For \emph{piecewise-constant} value density functions, a recent result by \cite{DBLP:journals/ai/BuST23} shows that truthfulness is incompatible with EF, which holds even for two agents.
For the more restrictive \emph{piecewise-uniform} value density functions, truthful and EF mechanisms are known to exist~\cite{CL10,DBLP:conf/wine/AzizY14,bei2020truthful}.

We have seen that fairness and truthfulness are incompatible in most scenarios for the fair division problem.
To bypass the strong negative results for combining truthfulness and fairness, many relaxed notions of (dominant-strategy) truthfulness are proposed, including \emph{ex-ante truthfulness}~\cite{DBLP:journals/ai/BeiTWY25,DBLP:conf/focs/BuT25}, \emph{maximin strategy-proofness}~\cite{brams2006better}, \emph{non-obviously manipulability}~\cite{troyan2020obvious,DBLP:conf/nips/0001V22,ortega2022obvious}, and \emph{risk-averse truthfulness}~\cite{DBLP:journals/ai/BuST23,DBLP:journals/corr/abs-2502-18805}.
Another series of research assumes goods to have \emph{binary} marginal values, in which case truthfulness is achievable together with various fairness and efficiency properties \cite{DBLP:conf/wine/0002PP020,DBLP:conf/aaai/BabaioffEF21,DBLP:conf/aaai/BarmanV22}.
\section{Preliminaries}
\label{sec:prelim}

There are $m$ goods to be allocated to $n$ agents.
We assume that each agent $a \in [n]$ is associated with a non-negative value $v_a(g)$ for every good $g \in [m]$.
These values induce an ordinal preference $\succ_a$, which is a strict ordering over $[m]$ such that $v_a(g) > v_a(g')$ implies $g \succ_i g'$ for all $g, g' \in [m]$.
Let $\bfv = (v_1,\ldots,v_n)$ be a valuation profile.
For a subset of agents $C\subseteq [n]$, denote $\bfv_{C} = (v_a)_{a \in C}$ as the valuation profile induced by agents in $C$, and denote $\bfv_{-C} = (v_a)_{a \in [n] \setminus C}$ as the valuation profile induced by agents in $[n] \setminus C$.
We consider both indivisible and divisible goods in this paper, which we separately introduce below.

\paragraph{Indivisible goods.}
When all goods are indivisible, a \emph{bundle} is a subset of $[m]$.
An \emph{integral allocation} $A = (A_1, \ldots, A_n)$ is a partition of $[m]$ satisfying $A_a \cap A_b = \emptyset$ for all $a \neq b$ and $\bigcup_{a=1}^n A_a = [m]$, where $A_a$ denotes the bundle received by agent $a$.
For an agent $a \in [n]$ and a subset of goods $S \subseteq [m]$, let $v_a(S)$ be the value of $S$ for agent $a$.
We assume that every valuation function $v_a$ is \emph{additive}, i.e., $v_a(S \cup T) = v_a(S) + v_a(T)$ for all $S, T \subseteq [m]$ with $S \cap T = \emptyset$.
The \emph{utility} of agent $a$ under an allocation $A$ is given by $v_a(A_a)$.

\paragraph{Divisible goods.}
When all goods are divisible, a \emph{bundle} is denoted by a vector $\bfx \in \R_{\geq 0}^m$, where $x_g$ denotes the fraction of good $g \in [m]$ contained in the bundle.
A \emph{fractional allocation} $\bfx = (\bfx_1, \ldots, \bfx_n) \in \R_{\geq 0}^{n \times m}$, where $\bfx_a$ denotes the bundle received by agent $a$, satisfies $\sum_{a=1}^n x_{a, g} = 1$ for every $g \in [m]$.
The value of agent $a$ for a bundle $\bfx_a$ is $v_a(\bfx_a) = \sum_{g=1}^m v_a(g) \cdot x_{a, g}$.
The \emph{utility} of agent $a$ under a fractional allocation $\bfx$ is given by $v_a(\bfx_a)$.

\paragraph{Mechanisms.}
A \emph{mechanism} takes agents' preferences as input and outputs an allocation.
Formally, a mechanism $\calM$ takes a valuation profile $\bfv = (v_1, \ldots, v_n)$ as input and outputs an allocation $\calM(\bfv) = (\calM_1(\bfv), \ldots, \calM_n(\bfv))$, where $\calM_a(\bfv)$ denotes the bundle received by agent $a$.
The allocation $\calM(\bfv)$ can be integral or fractional, depending on whether the goods are indivisible or divisible.
Agents' preferences are private information, and they will try to maximize their utilities by misreporting.

We will sometimes consider \emph{ordinal mechanisms}, whose outputs only depend on agents' ordinal preferences over goods.
In this case, the mechanism only asks each agent to submit an ordinal preference over goods instead of a valuation function.
Hence, the mechanism takes an ordinal profile $\succ = (\succ_1, \ldots, \succ_n)$ as input and outputs an allocation $\calM(\succ) = (\calM_1(\succ), \ldots, \calM_n(\succ))$.

\subsection{Group Incentive Ratio}

In this subsection, we define several notions to quantify agents' gains due to their strategic behaviors.
We first recall the definition of incentive ratio, which quantifies the gain of unilateral manipulation, and then generalize it to the setting where agents can collectively manipulate.
We only define them for cardinal mechanisms which take a valuation profile as input, and the definitions for ordinal mechanisms are analogous.

\begin{definition}[Incentive Ratio]\label{def: IR}
    The \emph{incentive ratio} of a mechanism $\calM$, denoted as $\ic_{\calM}$, is defined as the smallest value $R \geq 1$ such that for all agent $a \in [n]$ and manipulated valuation function $v_a'$,
    \begin{align*}
        v_a(\calM_i(v_a', \bfv_{-a}))
        \leq R \cdot v_a(\calM_i(\bfv)).
    \end{align*}
\end{definition}

If the incentive ratio of a mechanism is $1$, then we say that this mechanism is \emph{truthful}.
Next, we consider the case where agents can collaborate to manipulate.
We refer to the agents who form a coalition to manipulate as \emph{corrupted agents}, and the remaining agents who always report truthfully as \emph{honest agents}.
The first generalization of incentive ratio is defined as the maximum incentive ratio among the corrupted agents with restriction to the manipulated profiles where each corrupted agent is weakly better off.

\begin{definition}[Strong Group Incentive Ratio]\label{def:strong-gir}
    The \emph{strong group incentive ratio (SGIR)} of a mechanism $\calM$ against coalitions of size $c \geq 1$, denoted as $\gics_{\calM}(c)$, is defined as the smallest value $R \geq 1$ such that for all coalition $C \subseteq [n]$ and manipulated valuation functions $\bfv_C'$ with $|C| \leq c$, if $v_a(\calM_a(\bfv)) \leq v_a(\calM_a(\bfv_C', \bfv_{-C}))$ for every corrupted agent $a \in C$, then
    \begin{align*}
        v_a(\calM_a(\bfv_C', \bfv_{-C}))
        \leq R \cdot v_a(\calM_a(\bfv)), \quad \forall a \in C.
    \end{align*}
\end{definition}

The next generalization of incentive ratio quantifies the gain of the \emph{worst} corrupted agent instead of the best one.

\begin{definition}[Group Incentive Ratio]\label{def:gir}
    The \emph{group incentive ratio (GIR)} of a mechanism $\calM$ against coalitions of size $c \geq 1$, denoted as $\gic_{\calM}(c)$, is defined as the smallest value $R \geq 1$ such that for all coalition $C \subseteq [n]$ and manipulated valuation functions $\bfv_C'$ with $|C| \leq c$, there exists a corrupted agent $a \in C$ such that
    \begin{align*}
        v_a(\calM_a(\bfv_C', \bfv_{-C}))
        \leq R \cdot v_a(\calM_a(\bfv)).
    \end{align*}
\end{definition}

It is not hard to verify that the SGIR of a mechanism always upper bounds its GIR.

\begin{corollary}\label{cor:compare-defs-gic}
    For all mechanism $\calM$ and $c \geq 1$, $\gics_{\calM}(c) \geq \gic_{\calM}(c)$.
\end{corollary}

To see that both SGIR and GIR generalize incentive ratio, note that $\gics_{\calM}(1) = \gic_{\calM}(1) = \ic_{\calM}$ for every mechanism $\calM$.
Moreover, the definitions of SGIR and GIR resemble those of \emph{(weak) group strategyproofness (GSP)}~\cite{DBLP:conf/wine/AzizY14} and \emph{(super) strong Nash equilibrium (SNE)}~\cite{DBLP:journals/ijgt/BraggionGLSS20,DBLP:conf/sigecom/HanSTX23}.
In particular, $\gics_{\calM}(n) = 1$ is equivalent to $\calM$ being GSP, and $\gic_{\calM}(n) = 1$ is equivalent to $\calM$ being weak GSP.
Regarding (super) SNE, slightly varied definitions are adopted in prior work.
Following \cite{DBLP:journals/ijgt/BraggionGLSS20}, $\gics_{\calM}(n) = 1$ implies that truthful telling forms a super SNE, and $\gic_{\calM}(n) = 1$ implies that truthful telling forms a SNE.
\section{Maximum Nash Welfare}

In this section, we consider divisible goods and establish the SGIR and GIR of the Maximum Nash Welfare (MNW) mechanism.
Recall that the \emph{Nash welfare} of an allocation $\bfx$ is defined as the geometric mean of agents' utilities, i.e., $(\prod_{a=1}^n v_a(x_a))^{1 / n}$, and MNW returns an allocation that maximizes Nash welfare, breaking ties arbitrarily.
Assume without loss of generality that each agent has a strictly positive value for at least one good, which ensures the existence of an allocation with strictly positive Nash welfare.
Our main theorem in this section is the following.

\begin{theorem}\label{thm:gic-mnw}
    For every $c \geq 1$, $\gics_{\mnw}(c) = c + 1$ and $\gic_{\mnw}(c) = 2$.
\end{theorem}

The upper and lower bounds in \Cref{thm:gic-mnw} are proved in Sections~\ref{sec:proof-gic-mnw-ub} and~\ref{sec:proof-gic-mnw-lb}, respectively.

\subsection{Proof of \Cref{thm:gic-mnw}: Upper Bounds}
\label{sec:proof-gic-mnw-ub}

In this subsection, we show that $\gics_{\mnw}(c) \leq c + 1$ and $\gic_{\mnw}(c) \leq 2$.
It is known that the set of allocations maximizing Nash welfare coincides with the set of market equilibria in the corresponding \emph{Fisher market}~\cite[Chapter 5]{DBLP:books/cu/NRTV2007}, and we will prove the upper bounds by analyzing these market equilibria.

We first introduce the notion of market equilibrium.
In a Fisher market, there are still $n$ agents competing for $m$ divisible goods. Each agent has a budget in addition to their preferences. It is sufficient for us to consider the case where all agents have the same budgets, assumed to be $1$ by normalization.
An outcome of the market is denoted by a pair $(\bfx, \bfp)$, where $\bfx$ is an allocation and $\bfp \in \R_{\geq 0}^m$ is a \emph{price vector} with $p_g$ indicating the price of good $g$.
Given a price vector $\bfp$, an \emph{optimal allocation} $\bfx_a$ for an agent $a$ attains the maximum utility $v_a(\bfx_a)$ subject to the budget constraint, i.e., $\bfp\cdot \bfx_a\le 1$.
We use $D_a(\bfp)$ to denote the set of all optimal allocations for $a$ with respect to $\bfp$.
We define the market equilibrium as follows.

\begin{definition}[Market Equilibrium]
    We say that an outcome $(\bfx, \bfp)$ of the market forms a \emph{market equilibrium} if
    \begin{itemize}
        \item all goods are completely allocated: $\sum_{a=1}^n x_{a,g} = 1$ for every good $g \in [m]$,
        \item each agent spends all his budget: $\bfp \cdot \bfx_a = 1$ for every agent $a \in [n]$, and
        \item $\bfx_a \in D_a(\bfp)$ for every agent $a \in [n]$.
    \end{itemize}
\end{definition}

To prove the SGIR and GIR of MNW in our setting where the valuation function of every agent is additive, we consider the Fisher market where the valuation functions of all agents belong to the class of Constant Elasticity of Substitution functions.

\begin{definition}[Constant Elasticity of Substitution]
    A valuation function $v: \R_{\geq 0}^m \to \R_{\geq 0}$ is a \emph{Constant Elasticity of Substitution (CES)} function for $\rho \in (-\infty, 1)$ with $\rho \neq 0$ if $v(\bfx) = (\sum_{g=1}^m u_g x_g^{\rho})^{1 / \rho}$ for some $u_1, \ldots, u_m \geq 0$.
\end{definition}

Additive valuation functions are also CES functions.
As a result, if we can show an upper bound for the SGIR or GIR of MNW for CES valuation functions, then the same upper bound also holds for additive valuation functions in our divisible good setting.
Moreover, CES valuation functions are convenient to work with since they enjoy the \emph{weak gross substitute} property, which informally asserts that if we increase the prices of some goods, then the demand for any other good will not decrease (\Cref{lem:price-change-compare}).

Our goal now is to prove the following lemma, which directly implies the desired upper bounds for the SGIR and GIR of MNW.

\begin{lemma}\label{lem:gic-wgs-mnw}
    Suppose $v_1, \ldots, v_n$ are CES valuation functions, and the valuation functions misreported by corrupted agents are also CES.
    Then, for every $c \geq 1$, $\gics_{\mnw}(c) \leq c + 1$ and $\gic_{\mnw}(c) \leq 2$.
\end{lemma}

\begin{proof}
We start by presenting two properties of CES valuation functions.
The first property relates different bundles to their respective resulting utilities, which results from a simple application of strong duality and holds as long as the valuation function is concave.

\begin{lemma}[Proposition 2.1 in \cite{DBLP:journals/iandc/ChenDTZZ22}]\label{lem:bang-per-buck}
    Let $\bfx_a \in D_a(\bfp)$ be an optimal bundle for agent $a$ under the price vector $\bfp$.
    If $v_a$ is a CES valuation function, then there exists $r_a \in [0, v_a(\bfx_a)]$ such that for any other bundle $\bfx_a' \in \R_{\geq 0}^m$,
    \begin{align*}
        r_a \sum_{g=1}^m p_g (x_{a, g} - x_{a, g}')
        \leq v_a(\bfx_a) - v_a(\bfx_a').
    \end{align*}
\end{lemma}

The second property states that if prices are changed, then an agent will spend no less money on those goods whose prices are decreased, and no more money on all other goods.
The proof of the following lemma relies on the weak gross substitute property of CES valuation functions.

\begin{lemma}[Proposition 3.1 in \cite{DBLP:journals/iandc/ChenDTZZ22}]\label{lem:price-change-compare}
    Let $\bfp$ and $\bfp'$ be two price vectors, and let $\bfx \in D_a(\bfp)$ and $\bfx' \in D_a(\bfp')$ be an optimal bundle for agent $a$ under price vectors $\bfp$ and $\bfp'$, respectively.
    Let $S := \{g \in [m] \mid p_g > p_g'\}$ denote the set of goods whose prices are decreased from $\bfp$ to $\bfp'$.
    If $v_a$ is a CES valuation function, then
    \begin{align*}
        \sum_{g \in S} x_g p_g \leq \sum_{g \in S} x_g' p_g'
        \quad \text{and} \quad
        \sum_{g \in [m] \setminus S} x_g' p_g' \leq \sum_{g \in [m] \setminus S} x_g p_g.
    \end{align*}
\end{lemma}

Now, we are ready to prove \Cref{lem:gic-wgs-mnw}.
Fix the coalition size $c \geq 1$.
We relabel the agents so that agents $1, \ldots, c$ are corrupted.
Recall that the set of allocations maximizing Nash welfare coincides with the set of allocations corresponding to market equilibria~\cite[Chapter 5]{DBLP:books/cu/NRTV2007}.
Let $(\bfx, \bfp)$ be a market equilibrium when all agents report truthfully, and let $(\bfx', \bfp')$ be a market equilibrium when the corrupted agents misreport.

We can assume that $v_a(\bfx_a') \geq v_a(\bfx_a)$ for every corrupted agent $a \in [c]$ to fulfill the requirement by the definitions of SGIR and GIR.
We then claim that $\bfp \cdot \bfx_a' \geq \bfp \cdot \bfx_a = 1$. Assume not, when the price vector is $\bfp$, agent $a$ can buy a certain amount of some good on top of $\bfx_a'$ using his remaining budget $1 - \bfp \cdot \bfx_a'$ to get a bundle $\bfx_a''$ satisfying $v_a(\bfx_a'') > v_a(\bfx_a') \geq v_a(\bfx_a)$. As a result, $\bfx_a$ is not an optimal allocation for the agent $a$ as follows, which refutes that $(\bfx, \bfp)$ is a market equilibrium.

By \Cref{lem:bang-per-buck}, there exists $r_a \in [0, v_a(\bfx_a)]$ for every corrupted agent $a \in [c]$ such that
\begin{align*}
    v_a(\bfx_a') - v_a(\bfx_a)
    \leq r_a \sum_{g=1}^m p_g(x_{a, g}' - x_{a, g}).
\end{align*}
For ease of analysis, we scale $v_a$  and $r_a$ so that $v_a(\bfx_a) = 1$, and $r_a \leq v_a(\bfx_a) = 1$.
Note that scaling $v_a$ and $r_a$ does not change $D_a(\bfp)$ since $v_a$ is a CES valuation function. Hence, the market equibirium does not change, either.

We proceed to upper bound the total utility gain of misreporting for the group of corrupted agents.
Let $S := \{g \in [m] \mid p_g > p_g'\}$ denote the set of goods whose prices are decreased from $\bfp$ to $\bfp'$, and let $T := [m] \setminus S$ denote the set of goods whose prices are not decreased.
It follows that
\begin{align}
    \sum_{a=1}^c (v_a(\bfx_a') - v_a(\bfx_a))
    &\leq \sum_{a=1}^c r_a \sum_{g=1}^m p_g(x_{a, g}' - x_{a, g})
    \leq \sum_{a=1}^c \sum_{g=1}^m p_g(x_{a, g}' - x_{a, g}) \notag \\
    &= \sum_{a=1}^c \sum_{g \in S} p_g(x_{a, g}' - x_{a, g}) + \sum_{a=1}^c \sum_{g \in T} p_g(x_{a, g}' - x_{a, g}), \label{eqn:decom-s-t}
\end{align}
where the second inequality holds since $\bfp \cdot \bfx_a' \geq \bfp \cdot \bfx_a$ and $r_a \leq 1$ for every corrupted agent $a \in [c]$.

We show that the first term in \eqref{eqn:decom-s-t} is non-positive.
Since $\bfx_a \in D_a(\bfp)$ and $\bfx_a' \in D_a(\bfp')$ for every honest agent $a \in [c + 1: n]$, by \Cref{lem:price-change-compare}, he spends more money on goods in $S$ as their prices are decreased.
Hence,
\begin{align*}
    \sum_{g \in S} p_g - \sum_{a=1}^c \sum_{g \in S} x_{a, g} p_g
    = \sum_{a=c + 1}^n \sum_{g \in S} x_{a, g} p_g
    \leq \sum_{a=c + 1}^n \sum_{g \in S} x_{a, g}' p_g'
    = \sum_{g \in S} p_g' - \sum_{a=1}^c \sum_{g \in S} x_{a, g}' p_g',
\end{align*}
which implies
\begin{align*}
    \sum_{a=1}^c \sum_{g \in S} x_{a, g} p_g
    \geq \sum_{a=1}^c \sum_{g \in S} x_{a,g}' p_g' + \sum_{g \in S} (p_g - p_g').
\end{align*}
Therefore,
\begin{align*}
    \sum_{a=1}^c \sum_{g \in S} x_{a, g}' p_g - \sum_{a = 1}^c \sum_{g \in S} x_{a, g} p_g
    &\leq \sum_{a=1}^c \sum_{g \in S} x_{a, g}' p_g - \sum_{a=1}^c \sum_{g \in S} x_{a, g}' p_g' - \sum_{g \in S} (p_g - p_g') \\
    &= \sum_{g \in S} \left( \sum_{a=1}^c x_{a, g}' - 1 \right) (p_g - p_g')
    \leq 0,
\end{align*}
where the last inequality holds since $p_g > p_g'$ for every $g \in S$.
This concludes that the first term in \eqref{eqn:decom-s-t} is non-positive.

Next, we bound the second term in \eqref{eqn:decom-s-t} by
\begin{align*}
    \sum_{a=1}^c \sum_{g \in T} p_g (x_{a, g}' - x_{a, g})
    \leq \sum_{a=1}^c \sum_{g \in T} p_g x_{a, g}'
    \leq \sum_{a=1}^c \sum_{g \in T} p_g' x_{a, g}'
    \leq c,
\end{align*}
where the second inequality holds since $p_g \leq p_g'$ for every $g \in T$, and the last inequality holds since the money spent by each agent in the market equilibrium $(\bfx', \bfp')$ does not exceed his budget $1$.

Putting everything together, we derive that the total gain of the collusion is upper bounded by
\begin{align*}
    \sum_{a=1}^c (v_a(\bfx_a') -v_a(\bfx_a)) \leq c.
\end{align*}
Recall that $v_1, \ldots, v_c$ are scaled so that $v_1(\bfx_1) = \ldots = v_c(\bfx_c) = 1$.
Therefore, the GIR is maximized when the gain is split equally among the corrupted agents, which implies $\gic_{\mnw}(c) \leq 2$, and the SGIR is maximized when all gain is assigned to one corrupted agent, which implies $\gics_{\mnw}(c) \leq c + 1$.
This completes the proof.
\end{proof}

\subsection{Proof of \Cref{thm:gic-mnw}: Lower Bounds}
\label{sec:proof-gic-mnw-lb}

In this subsection, we show that $\gics_{\mnw}(c) \geq c + 1$ and $\gic_{\mnw}(c) \geq 2$ for every $c \geq 1$.

    We first show the lower bound for $\gic_{\mnw}(c)$.
    Assume that there are $n > c$ agents with $n$ being sufficiently large, and there are $m = n$ goods.
    The valuation functions are defined as follows.
    For every agent $a \in [c]$, $v_a(g) = 1$ for every good $g \in [m]$.
    For every agent $a \in [c + 1 : n]$,
    \begin{align*}
        v_a(g) =
        \begin{cases}
            0, & g \in [c], \\
            1, & g \in [c + 1: n].
        \end{cases}
    \end{align*}
    One possible allocation $\bfx$ returned by MNW satisfies that good $g$ is entirely allocated to agent $a$ for all $g, a \in [n]$ with $g = a$.
    Under this allocation, the utility of each agent is $1$.

    Now, assume that agents $1, \ldots, c$ form a coalition and collectively manipulate.
    In particular, agent $a \in [c]$ misreports his valuation function as $v_a'$ satisfying
    \begin{align*}
        v_a'(g) =
        \begin{cases}
            0, & g \in [c],\\
            1, & g \in [c + 1: n].
        \end{cases}
    \end{align*}
    In an allocation $\bfx'$ that maximizes Nash welfare with respect to the manipulated profile, since no agent values goods $1, \ldots, c$, they are allowed to be allocated arbitrarily, and we assume that good $g$ is entirely allocated to agent $a$ for all $a, g \in [c]$ with $g = a$.
    For the remaining goods $c + 1, \ldots, n$, since each agent has value $1$ for each of them, the resulting allocation $\bfx'$ must satisfy $\sum_{g=c + 1}^n x_{a, g}' = (n - c) / n$ for every agent $a \in [n]$, i.e., these goods are allocated to all agents uniformly.
    For each corrupted agent $a \in [c]$, his utility with respect to his true preference becomes $v_a(\bfx_a') = 1 + (n - c) / n$, and hence
    \begin{align*}
        \frac{v_a(\bfx_a')}{v_a(\bfx_a)}
        = 2 - \frac{c}{n},
    \end{align*}
    which approaches $2$ as $n \to \infty$.
    Therefore, $\gic_{\mnw}(c) \geq 2$.

    Next, we show the lower bound for $\gics_{\mnw}(c)$.
    For the same instance as above, if we modify the allocation of goods $1, \ldots, c$ in $\bfx'$ so that $\sum_{g=1}^c x_{1, g}' = c(1 - 1 / n)$ and $\sum_{g=1}^c x_{a, g}' = c / n$ for every agent $a \in [2: c]$, then it holds that $v_a(\bfx_a') = 1 = v_a(\bfx_a)$ for every agent $a \in [2: c]$, and
    \begin{align*}
        \frac{v_1(\bfx_1')}{v_1(\bfx_1)}
        = 1 + c - \frac{c}{n},
    \end{align*}
    which approaches $c + 1$ as $n \to \infty$.
    Therefore, $\gics_{\mnw}(c) \geq c + 1$.
\section{Reducing Probabilistic Serial to Round-Robin}

In this section, we consider the Round-Robin (RR) and Probabilistic Serial (PS) mechanisms, both of which are ordinal mechanisms and hence take an ordinal profile as input.
We show that the SGIR and GIR of PS are upper bounded respectively by the SGIR and GIR of RR, which will be utilized to prove the SGIRs and GIRs of both mechanisms.

RR is designed for indivisible goods and works as follows.
In each round, agents alternately select their favorite good among those that still remain unselected.
In particular, we assume that agents with a smaller index always select before agents with a larger index in each round.
We call the process in which an agent picks a good as a \emph{stage}, and there are $m$ stages in total.

Next, we introduce PS, which is defined for allocating divisible goods.
Given an ordering profile as input, each agent starts eating his favorite good at unit speed.
Whenever a good is completely consumed, the agents who are eating it continue to eat their favorite goods that are still available.
The mechanism terminates after all goods are completely consumed, and the fraction of each good received by each agent is proportional to the amount of time that he has spent on eating it.

Although the description of PS is continuous, it admits a polynomial-time combinatorial implementation (see, e.g., \cite{DBLP:conf/wine/Aziz20}), which we present in Mechanism~\ref{alg:probabilistic-serial}.
Given a subset of goods $S \subseteq [m]$ with $S \neq \emptyset$ and a strict ordering $\succ$ over all goods, we use $\max_{\succ} (S)$ to denote the most favorable good in $S$ with respect to $\succ$.
For all subset of goods $S \subseteq [m]$ and good $g \in S$, let $N(g, S) := \{a \in [n] \mid g=\max_{\succ_a} (S)\}$ be the set of agents whose favorite good in $S$ is $g$.
For every subset of goods $S \subseteq [m]$, let $\max_N(S) := \{g \in S \mid \exists a \in [n] \text{ s.t. } g = \max_{\succ_a} (S)\}$ denote the set of goods in $S$ which serves as at least one agent's favorite good in $S$.
On a high level, in Mechanism~\ref{alg:probabilistic-serial}, we iteratively compute the finishing time of each good in order.

\begin{algorithm}[t]
    \caption{Probabilistic Serial}
	\label{alg:probabilistic-serial}
    \begin{algorithmic}[1]
        \Require {An ordinal profile $\succ$}
        \Ensure {A fractional allocation $\bfx$}
        \State {$k \leftarrow 0$} \Comment{$k$ is the step of the mechanism.}
        \State {$S^{(0)} \leftarrow [m]$; $t^{(0)} \leftarrow 0$; $x_{a, g}^{(0)} \leftarrow 0$ for all $a \in [n]$ and $g \in [m]$} \Comment{$S^{(k)}$ is the set of available goods at the end of step $k$, and $\bfx^{(k)}$ is the partial allocation at the end of step $k$.}
        \While {$S^{(k)} \neq \emptyset$}
            \For {$g \in \max_N(S^{(k)})$}
                \State {$t^{(k + 1)}(g) \leftarrow \frac{1 - \sum_{a=1}^n x^{(k)}_{a, g}}{|N(g, S^k)|}$} \Comment{The finishing time of each good that is being eaten.}
            \EndFor
            \State {$t^{(k + 1)} \leftarrow \min_{g \in \max_N(S^{(k)})} t^{(k + 1)}(g)$} \Comment{The minimum time point in which some good is completely consumed.}
            \For {all $a \in [n]$ and $g \in [m]$}
                \If {$a \in N(g, S^{(k)})$}
                    \State {$x_{a, g}^{(k + 1)} \leftarrow x_{a, g}^{(k)} + t^{(k + 1)}$}
                \Else
                    \State {$x_{a, g}^{(k + 1)} \leftarrow x_{a, g}^{(k)}$}
                \EndIf
            \EndFor
            \State {$S^{(k + 1)} \leftarrow S^{(k)} \setminus \{g \in \max_N(S^{(k)}) \mid t^{(k + 1)}(g) = t^{(k + 1)}\}$} \Comment{Remove the finished goods.}
            \State {$k \leftarrow k + 1$}
        \EndWhile
        \State {\Return $\bfx = \bfx^{(k)}$}
    \end{algorithmic}
\end{algorithm}

The main result in this section is the following theorem.

\begin{theorem}\label{thm:gics-ps-grr}
    For every $c \geq 1$, $\gics_{\ps}(c) \leq \gics_{\rr}(c)$ and $\gic_{\ps}(c) \leq \gic_{\rr}(c)$.
\end{theorem}

We first show that PS can be reduced to RR in \Cref{sec:reduce-ps-rr}, and then we leverage this reduction to prove \Cref{thm:gics-ps-grr} in \Cref{sec:proof-gics-ps-grr}.

\subsection{Reduction}
\label{sec:reduce-ps-rr}

In this subsection, we show that PS can be reduced to RR.
Specifically, given an instance with divisible goods, we will transform it into an instance with indivisible goods by performing a sufficiently fine-grained discretization of each divisible good.
That is, we divide each divisible good into sufficiently many \emph{indivisible} copies and run RR on the new instance with indivisible goods.
Then, we show that the fractional allocation induced by the outcome of RR recovers the allocation returned by PS.
The implementation of PS via RR is formally described in Mechanism~\ref{alg:probabilistic-serial-by-round-robin}.
We remark that even though our reduction is inefficient, time complexity is not the primary concern here, and the reduction only serves a characterization purpose.

\begin{algorithm}[t]
    \caption{Implementation of Probabilistic Serial via Round-Robin}
	\label{alg:probabilistic-serial-by-round-robin}
    \begin{algorithmic}[1]
        \Require {An ordinal profile $\succ$}
        \Ensure {A fractional allocation $\bfx$}
        \State {$T \leftarrow (n!)^m$; $G(g) \leftarrow \{g_k \mid k \in [T]\}$ for every $g \in [m]$; $G \leftarrow \bigcup_{g=1}^m G(g)$ } \label{code:ps-rr-init}
        \State {Let $\succ'$ be an ordering profile over $G$ satisfying for all $a \in [n]$ and $g_k, g_{k'}' \in G$, $g_k \succ_i' g_{k'}'$ if and only if $g \succ_i g'$ or $g = g'$ and $k < k'$ } \label{code:ps-rr-order}
        \State {$A \leftarrow \rr(\succ')$}
        \For {all $a \in [n]$ and $g \in [m]$}
            \State {$x_{a, g} \leftarrow \frac{|A_a \cap G(g)|}{T}$}
        \EndFor
        \State {\Return $\bfx$}
    \end{algorithmic}
\end{algorithm}

The following theorem asserts that Mechanisms~\ref{alg:probabilistic-serial} and~\ref{alg:probabilistic-serial-by-round-robin} are equivalent in the sense that they always return the same allocations.

\begin{theorem}\label{thm:equi-ps-rr}
    Mechanisms~\ref{alg:probabilistic-serial} and~\ref{alg:probabilistic-serial-by-round-robin} are equivalent.
\end{theorem}

The rest of this subsection is devoted to proving \Cref{thm:equi-ps-rr}.
Given an ordinal profile $\succ$, let $\bfx = \ps(\succ)$ be the fractional allocation returned by PS, and let $A = \rr(\succ')$ be the integral allocation returned by $\rr$ on ordinal profile $\succ'$ defined in Mechanism~\ref{alg:probabilistic-serial-by-round-robin}.
\Cref{thm:equi-ps-rr} is equivalent to
\begin{align}\label{eqn:equi-ps-rr-per-agent-good}
    x_{a, g} = \frac{|A_a \cap G(g)|}{T}
\end{align}
for all $a \in [n]$ and $g \in [m]$.
We first characterize the execution of Mechanism~\ref{alg:probabilistic-serial} on input $\succ$.

Assume that Mechanism~\ref{alg:probabilistic-serial} terminates after $K$ steps, and note that $K \leq m$ since at least one good is completely consumed after each step.
We first show that $(n!)^k$ is a multiple of the denominator of $t^{(k)}$ for every $k \leq K$, i.e., $t^{(k)} \cdot (n!)^k$ is an integer, where $t^{(k)}$ is defined in Mechanism~\ref{alg:probabilistic-serial}.

\begin{lemma}\label{lem:int-tk-dot-factnk}
    For every $k \in [0: K]$, $t^{(k)} \cdot (n!)^k$ is an integer.
\end{lemma}

\begin{proof}
    We inductively prove that $t^{(k)} \cdot (n!)^k$ is an integer together with the statement that $x_{a, g}^{(k)} \cdot (n!)^k$ is an integer for all $a \in [n]$ and $g \in [m]$, where $x_{a, g}^{(k)}$ is defined in Mechanism~\ref{alg:probabilistic-serial}.
    When $k = 0$, the statements straightforwardly hold since $t^{(0)} = 0$ and $x_{a, g}^{(0)} = 0$ for all $a \in [n]$ and $g \in [m]$.

    Assume for induction that the statements hold for $k = h - 1$ with $h \in [K]$, and we show that the statements also hold for $k = h$.
    Note that
    \begin{align*}
        t^{(h)}
        = t^{(h)}(g')
        = \frac{1 - \sum_{a=1}^n x_{a, g'}^{(h - 1)}}{|N(g', S^{(h - 1)})|}
    \end{align*}
    for some $g' \in \max_N(S^{(h - 1)})$ satisfying $|N(g', S^{(h - 1)})| \in [n]$.
    Since $x_{a, g'}^{(h - 1)} \cdot (n!)^{h - 1}$ is an integer for every $a \in [n]$ by the inductive hypothesis, $(1 - \sum_{a=1}^n x_{a, g'}^{(h - 1)}) \cdot (n!)^{h - 1}$ is also an integer.
    Moreover, since $|N(g', S^{(h - 1)})| \in [n]$, it follows that
    \begin{align*}
        t^{(h)} \cdot (n!)^h
        = \frac{1 - \sum_{a=1}^n x_{a, g'}^{(h - 1)}}{|N(g', S^{(h - 1)})|} \cdot (n!)^h
        = \left( 1 - \sum_{a=1}^n x_{a, g'}^{(h - 1)} \right) \cdot (n!)^{h - 1} \cdot \frac{n!}{|N(g', S^{(h - 1)})|}
    \end{align*}
    is an integer as well.
    Finally, for all $a \in [n]$ and $g \in [m]$, since both $x_{a, g}^{(h - 1)} \cdot (n!)^h$ and $t^{(h)} \cdot (n!)^h$ are integers,
    \begin{align*}
        x_{a, g}^{(h)} \cdot (n!)^h
        = x_{a, g}^{(h - 1)} \cdot (n!)^h + \mathbf{1}\{a \in N(g, S^{(h - 1)})\} \cdot t^{(h)} \cdot (n!)^h
    \end{align*}
    is also an integer, concluding the proof.
\end{proof}

Next, the following lemma provides two equivalent expressions for the total assigned fraction of each good during the first $h$ steps of Mechanism~\ref{alg:probabilistic-serial}.

\begin{lemma}\label{lem:two-form-of-assigned-fraction}
    For all $h \in [0: K]$ and $g \in [m]$,
    \begin{align}\label{eqn:two-expre-assigned-frac}
        \sum_{k=1}^h t^{(k)} \cdot |N(g, S^{(k - 1)})|
        = \sum_{a=1}^n x^{(h)}_{a, g}.
    \end{align}
\end{lemma}

\begin{proof}
    Since $|N(g, S^{(k - 1)})|$ denotes the number of agents eating $g$ during step $k$ of Mechanism~\ref{alg:probabilistic-serial}, $t^{(k)} \cdot |N(g, S^{(k - 1)})|$ represents the total assigned fraction of $g$ during step $k$.
    Hence, the LHS of \eqref{eqn:two-expre-assigned-frac} equals the total assigned fraction of $g$ during the first $h$ steps, which coincides with the interpretation of the RHS of \eqref{eqn:two-expre-assigned-frac}.
\end{proof}

Now, we couple the executions of Mechanism~\ref{alg:probabilistic-serial} on input $\succ$ and $\rr(\succ')$.
Note that $\rr(\succ')$ consists of $R := mT / n$ rounds.
In particular, when step $k$ of Mechanism~\ref{alg:probabilistic-serial} is finished, we proceed $\rr(\succ')$ by $T \cdot t^{(k)}$ rounds.
The coupling is well-defined since $T \cdot t^{(k)} \geq 0$ is an integer by \Cref{lem:int-tk-dot-factnk} and the fact that $K \leq m$, and we have
\begin{align*}
    \sum_{k=1}^K T \cdot t^{(k)}
    = T \cdot \frac{m}{n}
    = R.
\end{align*}
Due to the coupling, for every $k \in [0: K]$, we call the execution of $\rr(\succ')$ from the beginning of round $T \cdot \sum_{i=1}^{k-1} t^{(i)} + 1$ to the end of round $T \cdot \sum_{i=1}^k t^{(i)}$ as step $k$ of $\rr(\succ')$.

For all $a \in [n]$ and $k \in [K]$, let $g_{k, a}$ denote the good eaten by agent $a$ during step $k$ of Mechanism~\ref{alg:probabilistic-serial}, i.e., $g_{k, a}$ denotes the good $g$ satisfying $a \in N(g, S^{(k - 1)})$.
We show that during step $k$ of $\rr(\succ')$, every agent $a$ receives $T \cdot t^{(k)}$ goods in $G(g_{k, a})$.
Since step $k$ of $\rr(\succ')$ consists of exactly $T \cdot t^{(k)}$ rounds, it further follows that agent $a$ does not receive any other goods during step $k$ of $\rr(\succ')$.

\begin{lemma}\label{lem:rr-per-step-allo}
    For all $a \in [n]$ and $k \in [0: K]$, agent $a$ receives $T \cdot t^{(k)}$ goods in $G(g_{k, a})$ during step $k$ of $\rr(\succ')$, where we define $G(g_{k, a}) = \emptyset$ for $k = 0$.
\end{lemma}

\begin{proof}
    We prove the lemma by induction on $k$.
    When $k = 0$, the statement straightforwardly holds since $t^{(0)} = 0$.
    Assume for induction that the statement holds for $k < h$ with $h \in [K]$, and we show that the statement also holds for $k = h$.
    At the beginning of step $h$ of $\rr(\succ')$, for every $g \in [m]$, the number of available goods in $G(g)$, denoted as $\ell_{h, g}$, equals
    \begin{align}\label{eqn:remain-avail-goods}
        \ell_{h, g}
        = T - \sum_{k=1}^{h - 1} T \cdot t^{(k)} \cdot |N(g, S^{(k - 1)})|
        = T \left( 1 - \sum_{a=1}^n x_{a, g}^{(h - 1)} \right),
    \end{align}
    where the first equality holds by the inductive hypothesis, and the second equality holds by \Cref{lem:two-form-of-assigned-fraction}.
    Furthermore, \eqref{eqn:remain-avail-goods} implies that for every $g \in [m]$, goods in $G(g)$ are completely consumed at the beginning of step $h$ of $\rr(\succ')$ if and only if $g$ is completely consumed at the beginning of step $h$ of Mechanism~\ref{alg:probabilistic-serial}.
    Hence, by the definition of $g_{h, a}$ and the construction of $\succ'$, each agent $a$ prefers goods in $G(g_{h, a})$ to all other goods in $G$ available at the beginning of step $h$ of $\rr(\succ')$.
    Recall that $S^{(h - 1)}$ denotes the set of available goods at the beginning of step $h$ of Mechanism~\ref{alg:probabilistic-serial}.
    As a result, to show that agent $a$ receives $T \cdot t^{(h)}$ goods in $G(g_{h, a})$ during step $h$ of $\rr(\succ')$, it suffices to show that $G(g)$ is not completely consumed before the end of step $h$ of $\rr(\succ')$ for every $g \in \max_N(S^{(h - 1)})$.
    
    Fix $g \in \max_N(S^{(h - 1)})$.
    Since $t^{(h)}(g) \geq t^{(h)}$,
    \begin{align}\label{ngsh-1-ub}
        |N(g, S^{(h - 1)})|
        = \frac{1 - \sum_{a=1}^n x_{a, g}^{(h - 1)}}{t^{(h)}(g)}
        \leq \frac{1 - \sum_{a=1}^n x_{a, g}^{(h - 1)}}{t^{(h)}}
        = \frac{\ell_{h, g}}{T \cdot t^{(h)}},
    \end{align}
    where the first equality holds by the definition of $t^{(h)}(g)$, and the last equality holds by \eqref{eqn:remain-avail-goods}.
    Hence, during step $h$ of $\rr(\succ')$, suppose that all agents in $N(g, S^{(h - 1)})$ keep consuming goods in $G(g)$, and it follows that the number of consumed goods in $G(g)$ equals
    \begin{align*}
        T \cdot t^{(h)} \cdot |N(g, S^{(h - 1)})|
        \leq T \cdot t^{(h)} \cdot \frac{\ell_{h, g}}{T \cdot t^{(h)}}
        = \ell_{h, g},
    \end{align*}
    where the inequality holds by \eqref{ngsh-1-ub}.
    Therefore, $G(g)$ contains a sufficient number of available good at step $h$ of $\rr(\succ')$, concluding the proof.
\end{proof}

Finally, we are ready to prove \Cref{thm:equi-ps-rr}.

\begin{proof}[Proof of \Cref{thm:equi-ps-rr}]
    Fix $a \in [n]$ and $g \in [m]$, and recall that our goal is to prove \eqref{eqn:equi-ps-rr-per-agent-good}.
    By the description of Mechanism~\ref{alg:probabilistic-serial},
    \begin{align}\label{eqn:formal-xig}
        x_{a, g}
        = \sum_{k=1}^K t^{(k)} \cdot \mathbf{1}\{g_{k, a} = g\}.
    \end{align}
    Moreover, by \Cref{lem:rr-per-step-allo},
    \begin{align}\label{eqn:formal-aiggt}
        \frac{|A_a \cap G(g)|}{T}
        = \frac{1}{T} \sum_{k=1}^K T \cdot t^{(k)} \cdot \mathsf{1}\{g_{k, a} = g\}
        = \sum_{k=1}^K t^{(k)} \cdot \mathbf{1}\{g_{k, a} = g\}.
    \end{align}
    Finally, \eqref{eqn:equi-ps-rr-per-agent-good} follows by combining \eqref{eqn:formal-xig} and \eqref{eqn:formal-aiggt}, which concludes the proof.
\end{proof}

\subsection{Proof of \Cref{thm:gics-ps-grr}}
\label{sec:proof-gics-ps-grr}

We only prove the statement for SGIR, and it can be established analogously for GIR.
Fix $c \geq 1$, and we aim to show that $\gics_{\ps}(c) \leq \gics_{\rr}(c)$.
Given a valuation profile $\bfv$ with a consistent ordering profile $\succ$, we construct an ordering profile $\succ'$ over a set of goods $G$ defined as in Lines~\ref{code:ps-rr-init} and~\ref{code:ps-rr-order} of Mechanism~\ref{alg:probabilistic-serial-by-round-robin}.
Moreover, we construct a valuation profile $\bfv'$ on $G'$ such that for all agent $a \in [n]$ and subset of goods $S \subseteq G$,
\begin{align*}
    v_a'(S)
    = \frac{1}{T} \sum_{g=1}^m v_a(g) \cdot |S \cap G(g)|,
\end{align*}
where $T$ and $G(g)$ are defined in Line~\ref{code:ps-rr-init} of Mechanism~\ref{alg:probabilistic-serial-by-round-robin}.
It is easy to verify that $\bfv'$ is an additive valuation profile.

The following lemma connects agents' utilities under different valuation/ordering profiles.

\begin{lemma}\label{lem:con-uti-ps-rr}
    For every subset of goods $S \subseteq G$, let $\bfx \in \R_{\geq 0}^m$ be a fractional bundle of $[m]$ satisfying $x_g = |S \cap G(g)| / T$ for every $g \in [m]$.
    Then, $v_a'(S) = v_a(\bfx)$ for every agent $a \in [n]$.
\end{lemma}

\begin{proof}
    By the definition of $v_a'$,
    \begin{align*}
        v_a'(S)
        = \frac{1}{T} \sum_{g=1}^m v_a(g) \cdot |S \cap G(g)|
        = \sum_{g=1}^m v_a(g) \cdot x_g
        = v_a(\bfx),
    \end{align*}
    as desired.
\end{proof}
    
Let $\bfx = \ps(\succ)$ be the fractional allocation returned PS on input $\succ$, and let $A = \rr(\succ')$ be the integral allocation returned by $\rr$ on input $\succ'$.
By Theorem~\ref{thm:equi-ps-rr}, $x_{a, g} = |A_a \cap G(g)| / T$ for all $a \in [n]$ and $g \in [m]$, and hence, by \Cref{lem:con-uti-ps-rr}, $v_a'(A_a) = v_a(x_a)$ for every agent $a \in [n]$.

Assume that agents in $C \subseteq [n]$ with $|C| \leq c$ form a coalition and misreport $\succ_C$ as $\widetilde{\succ}_C$ such that $v_a(\ps_a(\succ)) \leq v_a(\ps_a(\widetilde{\succ}_C, \succ_{-C}))$ for every corrupted agent $a \in C$, as required by the definition of SGIR.
For every $a \in C$, define a strict ordering $\widetilde{\succ}_a'$ over $G$ as in Line~\ref{code:ps-rr-order} of Mechanism~\ref{alg:probabilistic-serial-by-round-robin} with respect to $\widetilde{\succ}_a$.
Let $\widetilde{\bfx} = \ps(\widetilde{\succ}_C, \succ_{-C})$ be the fractional allocation returned by PS on the manipulated profile $(\widetilde{\succ}_C, \succ_{-C})$, and let $\widetilde{A} = \rr(\widetilde{\succ}_C', \succ_{-C}')$ be the integral allocation returned by $\rr$ on the manipulated profile $(\widetilde{\succ}_C', \succ_{-C}')$.
By \Cref{thm:equi-ps-rr}, $\widetilde{x}_{a, g} = |\widetilde{A}_a \cap G(g)| / T$ for all $a \in [n]$ and $g \in [m]$, and hence, by \Cref{lem:con-uti-ps-rr}, $v_a'(\widetilde{A}_a) = v_a(\widetilde{\bfx}_a)$ for every $a \in [n]$.
As a result, for every corrupted agent $a \in C$,
\begin{align*}
    v_a'(A_a)
    = v_a(\bfx_a)
    \leq v_a(\widetilde{\bfx}_a)
    = v_a'(\widetilde{A}_a).
\end{align*}
Therefore, by the definition of $\gics_{\rr}(c)$,
\begin{align*}
    \frac{v_a(\widetilde{\bfx}_a)}{v_a(\bfx_a)}
    = \frac{v_a'(\widetilde{A}_a)}{v_a'(A_a)}
    \leq \gics_{\rr}(c),
\end{align*}
which concludes that $\gics_{\ps}(c) \leq \gics_{\rr}(c)$.
\section{Round-Robin and Probabilistic Serial: Upper Bounds}
\label{sec:ub}

In this section, we first show an upper bound of $c + 1$ for $\gic_{\rr}(c)$.
Then, we establish the same upper bound of $c + 1$ for $\gics_{\ps}(c)$ using the upper bound for $\gic_{\rr}(c)$.

\subsection{Round-Robin}

\begin{theorem}\label{thm:gic-rr-ub}
    For every $c \geq 1$, $\gic_{\rr}(c) \leq c + 1$.
\end{theorem}

To prove \Cref{thm:gic-rr-ub}, we show in the following lemma that when agent $1$, who acts as the first agent to select in each round, is one of the corrupted agents, his utility cannot increase by a factor strictly larger than $c + 1$.

\begin{lemma}\label{lem:ic1-rr-group}
    For all $c \geq 1$ and set of corrupted agents $C \subseteq [n]$ with $|C| \leq c$ and $1 \in C$, for all manipulated ordinal preferences $\succ_C'$, we have $v_1(\rr_1(\succ_C', \succ_{-C})) \leq (c + 1) \cdot v_1(\rr_1(\succ))$.
\end{lemma}

The proof of \Cref{lem:ic1-rr-group} generalizes and simplifies the proof of \cite[Theorem 5.1]{tao2024fairtruthfulmechanismsadditive}, which is presented as follows.

\begin{proof}[Proof of \Cref{lem:ic1-rr-group}]
It is known that the maximum incentive ratio of an ordinal mechanism is attained when agents' valuation functions are \emph{binary}, i.e., $v_a(g) \in \{0, 1\}$ for each $g \in [m]$~\cite{DBLP:journals/jcss/HuangWWZ24}.
We show that this observation also holds for SGIR and GIR.

\begin{lemma}\label{lem:bi-ins}
    Given an ordinal mechanism $\calM$, fix a valuation profile $\bfv = (v_1, \ldots, v_n)$, a subset of agents $C \subseteq [n]$, and their manipulated ordinal preferences $\succ_C'$.
    Let $\succ$ be an ordering profile consistent with $\bfv$.
    For every agent $a \in C$, there exists a binary valuation function $\widehat{v}_a$ consistent with $\succ_a$ such that
    \begin{align*}
        \frac{\widehat{v}_a(\calM_a(\succ_C', \succ_{-C}))}{\widehat{v}_a(\calM_a(\succ))}
        \geq \frac{v_a(\calM_a(\succ_C', \succ_{-C}))}{v_a(\calM_a(\succ))}.
    \end{align*}
\end{lemma}

    The proof of \Cref{lem:bi-ins} is deferred to \Cref{sec:proof-bi-ins}.
    Since RR is an ordinal mechanism, by \Cref{lem:bi-ins}, it suffices to focus on binary valuation functions.
    
    For convenience, we use $g_1, \ldots, g_m$ to denote the goods.
    We renumber the goods so that for every $i \in [m]$, $g_i$ is the good received by some agent in stage $i$ under the true profile $\succ$, i.e., $g_i$ is the favorite good among all the remaining goods for the agent who is designated to receive a good in stage $i$.
    For every $i \in [0: m]$, define $G_i = \{g_{i + 1}, \ldots, g_m\}$ as the set of remaining goods at the end of stage $i$, and define $B_i$ as the set of goods received by agent $1$ until the end of stage $i$ under the true profile.
    Assume that the corrupted agents misreport their ordinal preferences as $\succ_C'$.
    Let $g_i'$ be the good received by some agent in stage $i$ under the manipulated profile $(\succ_C', \succ_{-C})$, and define $G_i' = \{g_{i + 1}', \ldots, g_m'\}$ and $B_i'$ analogously for $i \in [0: m]$.
    Note that $B_i' \cap G_i' = \emptyset$ for every $i \in [0: m]$.

    Denote $U = v_1(\rr_1(\succ))$ as the utility of agent $1$ under the true profile, which is also equal to the number of goods in $\rr_1(\succ)$ with value $1$ for agent $1$ since $v_1$ is binary.
    Observe that $v_1(g_i) = 0$ for all $i > nU$ since $v_1(g_{nU + 1}) = 0$ and $g_{nU + 1}$ is agent $1$'s favorite good among $g_{nU + 1}, \ldots, g_m$.
    Hence, define $\widehat{G} = \{g_1, \ldots, g_{nU}\}$, and we have $v_1(\rr_1(\succ_C', \succ_{-C})) \leq |\rr_1(\succ_C', \succ_{-C}) \cap \widehat{G}|$.
    Consequently, it suffices to show that
    \begin{align}\label{eqn:sz-rr1}
        |\rr_1(\succ_C', \succ_{-C}) \cap \widehat{G}|
        \leq (c + 1) \cdot v_1(\rr_1(\succ))
        = (c + 1) U.
    \end{align}
    
    We now describe how to keep track of the bundle received by agent $1$ under the manipulated profile.
    Intuitively, for each stage $i \in [m]$, $G_i' \setminus G_i$ contains all the goods in $\{g_1, \ldots, g_i\}$ ``left'' to the future stages by the corrupted agents via misreporting, and in the extreme case, they will all end up being received by the corrupted agents.
    For $i \in [0: m]$, let $X_i = (B_i' \cup (G_i' \setminus G_i)) \cap \widehat{G}$ be the set of goods in $(B_i' \cup \{g_1, \ldots, g_i\}) \cap \widehat{G}$ possibly received by agent $1$.
    We will show that $|X_m| \leq (c + 1) U$, which implies \eqref{eqn:sz-rr1} since $G_m' = G_m = \emptyset$ and $B_m' = \rr_1(\succ_C', \succ_{-C})$.

    We analyze how $|X_i|$ evolves in each stage.
    For stage $i \in [nU]$,
    \begin{itemize}
        \item \textbf{Case 1: stage $i$ is agent $1$'s turn.}
        $|X_i| - |X_{i - 1}| \leq 2$ since $|B_i' \setminus B_{i - 1}'| \leq 1$ and $|(G_i' \setminus G_i) \setminus (G_{i - 1}' \setminus G_{i - 1})| \leq 1$.

        \item \textbf{Case 2: stage $i$ is a corrupted agent $a \in C \setminus \{1\}$'s turn.}
        $|X_i| - |X_{i - 1}| \leq 1$ since $B_i' = B_{i - 1}'$ and $|(G_i' \setminus G_i) \setminus (G_{i - 1}' \setminus G_{i - 1})| \leq 1$.

        \item \textbf{Case 3: stage $i$ is an honest agent $a \in [n] \setminus C$'s turn.}
        We show $|X_i| \leq |X_{i - 1}|$ in this case.
        Note that $B_i' = B_{i - 1}'$, and hence $X_i \setminus X_{i - 1}$ is equal to either $\{g_i\}$ or $\emptyset$.
        If $g_i \notin G_i'$, then $X_i \subseteq X_{i - 1}$, which implies $|X_i| \leq |X_{i - 1}|$.
        On the other hand, if $g_i \in G_i'$, then we must have $g_i' \succ_a g_i$ since $g_i'$ is agent $a$'s favorite good in $G_{i - 1}' = G_i' \cup \{g_i'\}$.
        Moreover, since agent $a$ is honest, $g_i' \notin G_{i - 1}$ as otherwise, agent $a$ would have received $g_i'$ instead of $g_i$ in stage $i$ under the honest profile.
        As a result, $X_i = (X_{i - 1} \setminus \{g_i'\}) \cup \{g_i\}$, and hence $|X_i| = |X_{i - 1}|$.
    \end{itemize}
    For stage $i \in [nU+1: m]$, we claim that $|X_i| \leq |X_{i - 1}|$.
    To see this, for each $g \in \widehat{G}$, we show that $g \notin X_{i - 1}$ implies $g \notin X_i$.
    Note that $g \notin G_{i - 1}$ as $G_{i - 1} \cap \widehat{G} = \emptyset$.
    As a result, if $g \notin X_{i - 1}$, then $g \notin B_{i - 1}' \cup G_{i - 1}'$.
    Since agents can only receive goods in $G_{i - 1}'$ in stage $i$ under the manipulated profile, $G_i \subseteq G_{i - 1}$, and $G_i' \subseteq G_{i - 1}'$, it holds that $g \notin X_i$.
    This gives $|X_i| \leq |X_{i - 1}|$.
    
    To summarize the above arguments, $|X_i|$ increases by at most $2$ in agent $1$'s stages during the first $U$ rounds, increases by at most $1$ in each corrupted agent $a \in C \setminus \{1\}$'s stages during the first $U$ rounds, and does not increase otherwise.
    Hence,
    \begin{align*}
        |X_m|
        \leq 2U + (c - 1)U
        = (c + 1) U,
    \end{align*}
    which concludes the proof.    
\end{proof}

Now, we are ready to finish the proof of \Cref{thm:gic-rr-ub} by applying \Cref{lem:ic1-rr-group}.

\begin{proof}[Proof of \Cref{thm:gic-rr-ub}]
Note that if agent $1$ is not one of the corrupted agents $C$, then the corrupted agents' misreporting cannot alter the outcomes of RR in the first $\min_{a \in C} a - 1$ stages.
Hence, we can hypothetically remove the goods acquired by some agents in the first $\min_{a \in C} a - 1$ stages and pretend that RR starts from agent $\min_{a \in C} a$ in each round.
Now, the corrupted agent $\min_{a \in C} a$ plays the role of agent $1$, for whom \Cref{lem:ic1-rr-group} then applies.
\end{proof}

\subsection{Probabilistic Serial}

In this subsection, we upper bound the SGIR of PS.
Note that combining Theorems~\ref{thm:gics-ps-grr} and~\ref{thm:gic-rr-ub} only gives $\gic_{\ps}(c) \leq c + 1$.
Nevertheless, we manage to show that $\gics_{\ps}(c)$ is also upper bounded by $c + 1$ via leveraging the proof ingredients of both Theorems~\ref{thm:gics-ps-grr} and~\ref{thm:gic-rr-ub} more carefully.

\begin{theorem}\label{thm:gics-ps-ub}
    For every $c \geq 1$, $\gic_{\ps}(c) \leq \gics_{\ps}(c) \leq c + 1$.
\end{theorem}

\begin{proof}
    Due to \Cref{cor:compare-defs-gic}, it suffices to show that $\gics_{\ps}(c) \leq c + 1$.
    Fix $c \geq 1$.
    We will prove a stronger statement that no corrupted agent can improve his utility by a factor strictly larger than $c + 1$, even without the requirement that all corrupted agents must be weakly better off.
    Let $C \subseteq [n]$ with $|C| \leq c$ be the set of corrupted agents.
    Assume without loss of generality that $1 \in C$, and it suffices to show that agent $1$ cannot improve his utility by a factor strictly larger than $c + 1$.

    Given a valuation profile $\bfv$ with a consistent ordering profile $\succ$, we construct an ordering profile $\succ'$ over a set of goods $G$ defined as in Lines~\ref{code:ps-rr-init} and~\ref{code:ps-rr-order} of Mechanism~\ref{alg:probabilistic-serial-by-round-robin}.
    Moreover, we construct a valuation profile $\bfv'$ on $G$ such that for all agent $a \in [n]$ and subset of goods $S \subseteq G$,
    \begin{align*}
        v_a'(S)
        = \frac{1}{T} \sum_{g=1}^m v_a(g) \cdot |S \cap G(g)|,
    \end{align*}
    where $T$ and $G(g)$ are defined as in Line~\ref{code:ps-rr-init} of Mechanism~\ref{alg:probabilistic-serial-by-round-robin}.
    It is easy to verify that $\bfv'$ is an additive valuation profile.
    Let $\bfx = \ps(\succ)$ be the fractional allocation returned by PS on input $\succ$, and let $A = \rr(\succ')$ be the integral allocation returned by RR on input $\succ'$.
    By Theorem~\ref{thm:equi-ps-rr}, $x_{a, g} = |A_a \cap G(g)| / T$ for all $a \in [n]$ and $g \in [m]$.
    Hence, by \Cref{lem:con-uti-ps-rr}, $v_a'(A_a) = v_a(x_a)$ for every agent $a \in [n]$.
    
    Assume that the corrupted agents misreport $\succ_C$ as $\widetilde{\succ}_C$.
    For every $a \in C$, define a strict ordering $\widetilde{\succ}_a'$ over $G$ as in Line~\ref{code:ps-rr-order} of Mechanism~\ref{alg:probabilistic-serial-by-round-robin}.
    Let $\widetilde{\bfx} = \ps(\widetilde{\succ}_C, \succ_{-C})$ be the fractional allocation returned by PS on input $(\widetilde{\succ}_C, \succ_{-C})$, and let $\widetilde{A} = \rr(\widetilde{\succ}_C', \succ_{-C}')$ be the integral allocation returned by RR on input $(\widetilde{\succ}_C', \succ_{-C}')$.
    By \Cref{thm:equi-ps-rr}, $\widetilde{x}_{a, g} = |\widetilde{A}_a \cap G(g)| / T$ for all $a \in [n]$ and $g \in [m]$, and hence by \Cref{lem:con-uti-ps-rr}, $v_a'(\widetilde{A}_a) = v_a(\widetilde{x}_a)$ for every agent $a \in [n]$.
    Therefore, by \Cref{lem:ic1-rr-group},
    \begin{align*}
        \frac{v_1(\widetilde{x}_1)}{v_1(x_1)}
        = \frac{v_1'(\widetilde{A}_1)}{v_1'(A_1)}
        \leq c + 1,
    \end{align*}
    concluding the proof.
\end{proof}
\section{Round-Robin and Probabilistic Serial: Lower Bounds}
\label{sec:lb}

In this section, we prove lower bounds for the SGIRs and GIRs of PS and RR.
We first show that $\gic_{\ps}(c)$ is lower bounded by $c + 1$, which indicates that the upper bounds in \Cref{thm:gics-ps-ub} are tight.
Combining \Cref{thm:gics-ps-grr}, it follows that $\gic_{\rr}(c) \geq c + 1$, implying the upper bound in \Cref{thm:gic-rr-ub} is also tight.
Moreover, we show that $\gics_{\rr}(c)$ is unbounded for $c \geq 2$.

\subsection{Probabilistic Serial}

\begin{theorem}\label{thm:gic-ps-lb}
    For every $c \geq 1$, $\gics_{\ps}(c) \geq \gic_{\ps}(c) \geq c + 1$.
\end{theorem}

\begin{proof}
    Due to \Cref{cor:compare-defs-gic}, it suffices to show that $\gic_{\ps}(c) \geq c + 1$.
    Fix $c \geq 1$.
    Assume that there are $n$ agents and $m = (c + 1)nT^c$ divisible goods for sufficiently large $n > c$ and $T > 1$, where $T$ is a multiple of $n$.
    Let $C = [c]$ be the set of corrupted agents, and we label the goods by $g_a^{(p)}$ for $a \in [n]$ and $p \in [(c + 1) T^c]$.
    Let $G := \{g_a^{(p)} \mid a \in [n], p \in [(c + 1)T^c]\}$ denote the set of all goods.
    We first specify the valuation functions.
    For all $a \in [c]$ and $i \in [n]$,
    \begin{align*}
        v_a(g_i^{(p)}) =
        \begin{cases}
            1, & p \leq T^a,\\
            0, & p > T^a.
        \end{cases}
    \end{align*}
    For every $a \in [c + 1: n]$,
    \begin{align*}
        v_a(g_i^{(p)}) =
        \begin{cases}
            1, & i > c \text{ or } p > T^c,\\
            0, & i \leq c \text{ and } p \leq T^c.
        \end{cases}
    \end{align*}
    Next, we present each agent $a$'s ordinal preference $\succ_a$ consistent with $v_a$, and it is sufficient to describe how agent $a$ breaks ties among goods with the same value.
    On a high level, agent $a$ favors goods with a smaller superscript, subject to which he prioritizes goods with subscript $a$ and otherwise favors goods with a smaller subscript.
    More formally, for all agent $a \in [n]$ and goods $g_i^{(p)}, g_j^{(q)}$ with $g_i^{(p)} \neq g_j^{(q)}$ and $v_a(g_i^{(p)}) = v_a(g_j^{(q)})$, we have $g_i^{(p)} \succ_a g_j^{(q)}$ whenever at least one of the following conditions is met:
    \begin{enumerate}
        \item $p < q$;
        \item $p = q$ and $i = a$;
        \item $p = q$, $j \neq a$, and $i < j$.
    \end{enumerate}
    It is easy to verify that for all $a, i \in [n]$ and $p \in [(c + 1) T^c]$,
    \begin{align*}
        \ps_{a,g_i^{(p)}}(\succ)
        = \begin{cases}
            1, & i = a,\\
            0, & \text{otherwise}.
        \end{cases}
    \end{align*}
    In other words, for every agent $a \in [n]$, every good with subscript $a$ is completely received by agent $a$.
    Hence, $v_a(\ps_a(\succ)) = T^a$ for every corrupted agent $a \in [c]$.

    Before we specify the misreported ordinal preferences of the corrupted agents, we partition the goods into multiple sets.
    Let
    \begin{align*}
        G_0 = \left\{ g_i^{(p)} \; \Big \lvert \; i > c \text{ and } p \leq T^c \right\}, \qquad
        \widetilde{G} = \left\{g_i^{(p)} \; \Big \lvert \; p > M^c\right\},
    \end{align*}
    and for every $a \in [c]$, let
    \begin{align*}
        G_a = \left\{g_i^{(p)} \; \Big \lvert \; i \leq c \text{ and } T + T^2 + \ldots + T^{a - 1} < p \leq T^a\right\}.
    \end{align*}
    Note that $\widetilde{G} \cup G_0 \cup G_1 \cup \ldots \cup G_c = G$.
    Intuitively, since the goods in $G_0$ are valuable for both the corrupted agents and the honest agents, the corrupted agents will first compete for the goods in $G_0$; after goods in $G_0$ are all exhausted, each corrupted agent $a \in [c]$ will collect all the goods in $G_a$.
    Now, we specify the ordinal preferences $\succ_C'$ misreported by the corrupted agents.
    For all subsets of goods $G', G''$ with $G' \cap G'' = \emptyset$, we use $G' \succ G''$ to denote that $g' \succ g''$ for all $g' \in G'$ and $g'' \in G''$.
    For every corrupted agent $a \in [c]$, $\succ_a'$ satisfies
    \begin{align*}
        G_0 \succ_a' G_a \succ_a' \widetilde{G} \succ_a' G \setminus (G_0 \cup G_a \cup \widetilde{G}).
    \end{align*}
    In particular, among goods in $G_a$, agent $a$ breaks ties in favor of goods with a smaller superscript, subject to which he favors goods with a smaller subscript.
    Agent $a$'s tie-breaking rule among other goods can be arbitrary.

    It remains to analyze $\ps(\succ_C', \succ_{-C})$.
    At the beginning of the execution of $\ps(\succ_C', \succ_{-C})$, since the goods in $G_0$ have the top priority for all agents, they will exhaust the goods in $G_0$ in the first
    \begin{align*}
        \frac{|G_0|}{n}
        = \frac{(n - c)T^{c}}{n}
    \end{align*}
    units of time.
    Moreover, since all agents break ties among $G_0$ in favor of goods with a smaller superscript, we have
    \begin{align*}
        \sum_{i=c + 1}^n \ps_{a, g_i^{(p)}}(\succ_C', \succ_{-C}) = \frac{n - c}{n}
    \end{align*}
    for all $a \in [n]$ and $p \in [T^c]$.
    After all the goods in $G_0$ are exhausted, every honest agent starts receiving the goods in $\widetilde{G}$, and every corrupted agent $a \in [c]$ starts receiving the goods in $G_a$.
    To see that every corrupted agent $a \in [c]$ has enough time to receive all the goods in $G_a$, note that every agent will receive $(c + 1) T^c$ units of goods in total, and
    \begin{align*}
        \frac{|G_0|}{n} + |G_a|
        = \frac{(n - c)T^c}{n} + cT^a - \sum_{k=1}^{a - 1}cT^k
        < (c + 1) T^c.
    \end{align*}
    Hence, every corrupted agent $a \in [c]$ will first exhaust all the goods in $G_a$ and then keep receiving goods in $\widetilde{G}$ until the mechanism terminates.
    By the above analysis, for every corrupted agent $a \in [c]$,
    \begin{align*}
        v_a(\ps_a(\succ'_C, \succ_{-C}))&
        = \sum_{i = 1}^c \sum_{p = 1}^{T^a} \ps_{a, g_i^{(p)}}(\succ_C', \succ_{-C}) + \sum_{i = c + 1}^n \sum_{p = 1}^{T^a} \ps_{a, g_i^{(p)}}(\succ_C', \succ_{-C})\\
        &= |G_a| + \frac{|G_0|}{n}\\
        &= cT^a - \sum_{k=1}^{a - 1}cT^k + \frac{(n - c) T^a}{n}\\
        &= cT^a - c \cdot \frac{T^a - T}{T - 1} + \frac{(n - c)T^a}{n}.
    \end{align*}
    Therefore,
    \begin{align*}
        \gic_{\ps}(c)
        \geq \min_{a \in [c]} \frac{v_a(\ps_a(\succ'_C, \succ_{-C}))}{v_a(\ps_a(\succ))}
        = \min_{a \in [c]} \left( c - c \left( \frac{1}{T - 1} - \frac{1}{T^{a - 1} (T - 1)} \right) + \frac{n - c}{n} \right),
    \end{align*}
    which approaches $c + 1$ as $n \to \infty$ and $T \to \infty$.    
\end{proof}

The following is a direct corollary of Theorems~\ref{thm:gics-ps-grr} and~\ref{thm:gic-ps-lb}.

\begin{corollary}\label{cor:lb-gic-grr-eve}
    For every $c \geq 1$, $\gic_{\rr}(c) \geq c + 1$.
\end{corollary}

\subsection{Round-Robin}

Next, we show that $\gics_{\rr}(c)$ is unbounded for $c \geq 2$.
Note that for the case of $c = 1$, we have $\gics_{\rr}(1) = \ic_{\rr} = 2$.

\begin{theorem}\label{thm:gic-rr}
    For every $c \geq 2$, $\gics_{\rr}(c)$ is unbounded.
\end{theorem}

\begin{proof}
    It suffices to show that $\gics_{\rr}(2)$ is unbounded.
    Consider an instance with $n = 3$ agents and $m = 4$ goods.
    Let $\epsilon > 0$ be a sufficiently small real number.
    The valuation functions are given by
\begin{table}[H]
\centering
\begin{tabular}{@{}lllll@{}}
\toprule
      & $g_1$            & $g_2$          & $g_3$       & $g_4$       \\ \midrule
$v_1$ & $1 + 2 \epsilon$ & $1 + \epsilon$ & $1$         & 0 \\
$v_2$ & $1 / \epsilon$   & $1$            & $0$ & 0 \\
$v_3$ & $0$              & $0$            & $2$         & $1$         \\ \bottomrule
\end{tabular}
\end{table}
    \noindent Note that $\rr(\succ) = (\{g_1, g_4\}, \{g_2\}, \{g_3\})$, and thus $v_1(\rr_1(\succ)) = 1 + 2\epsilon$, and $v_2(\rr_2(\succ)) = 1$.

    Now, suppose that agents $1$ and $2$ form a coalition and misreport their ordinal preferences as $\succ_1'$ and $\succ_2'$ satisfying $\succ_2' = \succ_2$ and
    \begin{align*}
        g_3 \succ_1' g_2 \succ_1' g_1 \succ_1' g_4.
    \end{align*}
    It is easy to verify that $\rr(\succ_1', \succ_2', \succ_3) = (\{g_2, g_3\}, \{g_1\}, \{g_4\})$.
    As a result,
    \begin{align*} 
        v_1(\rr_1(\succ_1', \succ_2', \succ_3))
        = 2 + \epsilon
        > v_1(\rr_1(\succ)),
    \end{align*}
    and
    \begin{align*}
        \frac{v_2(\rr_2(\succ_1', \succ_2', \succ_3))}{v_2(\rr_2(\succ))}
        = \frac{1}{\epsilon}.
    \end{align*}
    Since $\epsilon$ can be arbitrarily small, we conclude that $\gics_{\rr}(2)$ is unbounded.
\end{proof}
\section{Discussion and Future Directions}

In this paper, we generalize the incentive ratio framework to capture collusive manipulation and tightly characterize the SGIRs and GIRs of MNW, PS, and RR.
Given that RR remains the only known fair mechanism for indivisible goods that admits a bounded incentive ratio, it is tempting to conjecture that any reasonably fair mechanism for indivisible goods admits an unbounded SGIR and a GIR that grows with the coalition size.
However, we expect this problem to be challenging as all current negative results for indivisible goods only hold for two agents~\cite{DBLP:conf/sigecom/AmanatidisBCM17,tao2024fairtruthfulmechanismsadditive}.
In particular, the resolution of the above conjecture would imply that there does not exist a fair and truthful mechanism when there are more than two agents, which still remains an open problem.

\section*{Acknowledgments}
The research of Biaoshuai Tao is supported by the National Natural Science Foundation of China (No. 62472271).

\bibliographystyle{alpha}
\bibliography{references}

\newpage
\appendix
\section{Proof of \Cref{lem:bi-ins}}
\label{sec:proof-bi-ins}

    Fix an agent $a \in C$.
    Assume without loss of generality that the ordinal preference of $a$ satisfies $1 \succ_a \cdots \succ_a m$.
    For every $g \in [m]$, let $\ell_g := \calM_{a, g}(\succ)$ and $\ell_g' := \calM_{a, g}(\succ_C', \succ_{-C})$ be the fractions of $g$ received by agent $a$ under the true and manipulated profiles, respectively.
    Hence,
    \begin{align*}
        \frac{v_a(\calM_a(\succ_C', \succ_{-C}))}{v_a(\calM_a(\succ))}
        = \frac{\ell_1' v_a(1) + \ldots + \ell_m' v_a(m)}{\ell_1 v_a(1) + \ldots + \ell_m v_a(m)}.
    \end{align*}
    Denote
    \begin{align*}
        g' \in \argmax_{g \in [0: m]} \frac{\ell_1' + \ldots + \ell_g'}{\ell_1 + \ldots + \ell_g}
        \qquad \text{and} \qquad
        r_{\max} := \frac{\ell_1' + \ldots + \ell_{g'}'}{\ell_1 + \ldots + \ell_{g'}},
    \end{align*}
    where we adopt the convention $0 / 0 = 1$.
    Note that
    \begin{align*}
        \frac{\ell_1' v_a(1) + \ldots + \ell_m' v_a(m)}{\ell_1 v_a(1) + \ldots + \ell_m v_a(m)}
        = \frac{\sum_{g=1}^m (v_a(g) - v_a(g + 1)) \sum_{i=1}^g \ell_i'}{\sum_{g=1}^m (v_a(g) - v_a(g + 1)) \sum_{i=1}^g \ell_i},
    \end{align*}
    where we denote $v_a(m + 1) = 0$.
    It follows that
    \begin{align*}
        r_{\max}
        \geq \frac{\ell_1' v_a(1) + \ldots + \ell_m' v_a(m)}{\ell_1 v_a(1) + \ldots + \ell_m v_a(m)}
        = \frac{v_a(\calM_a(\succ'_C, \succ_{-C}))}{v_a(\calM_a(\succ))}.
    \end{align*}

    Now, we define the desired binary valuation function $\widehat{v}_a$.
    Let $\widehat{v}_a(g) = 1$ for $g = 1, \ldots, g'$ and $\widehat{v}_a(g) = 0$ for $g = g' + 1, \ldots, m$.
    It is easy to verify that $\widehat{v}_a$ is consistent with $\succ_a$.
    As a result,
    \begin{align*}
        \frac{\widehat{v}_a(\calM_a(\succ_C', \succ_{-C}))}{\widehat{v}_a(\calM_a(\succ))}
        &= \frac{\ell_1' \widehat{v}_a(1) + \ldots + \ell_m' \widehat{v}_a(m)}{\ell_1 \widehat{v}_a(1) + \ldots + \ell_m \widehat{v}_a(m)}
        = \frac{\sum_{g=1}^m (\widehat{v}_a(g) - \widehat{v}_a(g + 1)) \sum_{i = 1}^g \ell_i'}{\sum_{g=1}^m (\widehat{v}_a(g) - \widehat{v}_i(g + 1)) \sum_{i=1}^g \ell_i}\\
        &= \frac{\sum_{i=1}^{g'} \ell_i'}{\sum_{i=1}^{g'} \ell_i}
        = r_{\max}
        \geq \frac{v_a(\calM_a(\succ'_C, \succ_{-C}))}{v_a(\calM_a(\succ))},
    \end{align*}
    which concludes the proof.

\end{document}